\documentclass[11pt]{article}
\usepackage[letterpaper,margin=1in]{geometry}
\usepackage[T1]{fontenc}
\usepackage{lmodern,microtype,amsmath,amssymb,amsthm,mathtools}
\usepackage{booktabs,tabularx,array,xcolor}
\usepackage[colorlinks=true,linkcolor=blue!45!black,citecolor=blue!45!black,urlcolor=blue!45!black]{hyperref}

\usepackage{authblk}
\usepackage{xcolor}
\usepackage{algorithm}
\usepackage{algpseudocode}
\usepackage{float}
\usepackage[colorlinks=true,linkcolor=blue!45!black,citecolor=blue!45!black,urlcolor=blue!45!black]{hyperref}
\usepackage{tikz}
 \usetikzlibrary{arrows.meta} 
 
\definecolor{figInk}{HTML}{222222}
\definecolor{figAttach}{HTML}{7B5EA7}
\definecolor{figBlue}{HTML}{2B6CB0}
\definecolor{figBad}{HTML}{E07A1F}
\definecolor{figGood}{HTML}{16877C}

\newtheorem{theorem}{Theorem}[section]
\newtheorem{lemma}[theorem]{Lemma}

\newtheorem{corollary}[theorem]{Corollary}

\theoremstyle{remark}

\newcommand{\poly}{\operatorname{poly}}

\begin{document}
\title{Polynomial-time additive-error estimation of output probabilities\\ for shallow quantum circuits}
\author[1]{\mbox{Matthew Coudron}}
\author[1, 2]{\mbox{Michael J. Gullans}}
\author[1]{\mbox{Jon Nelson}}
\author[1,3]{\authorcr\mbox{Joel Rajakumar}}
\author[1]{\mbox{Shi Jie Samuel Tan}}

\affil[1]{University of Maryland, College Park}
\affil[2]{QuEra Computing Inc.}
\affil[3]{University of California, Los Angeles}
\date{}
\maketitle
\begingroup
\renewcommand{\thefootnote}{}
\footnotetext{Corresponding authors:
\href{mailto:nelson1@umd.edu}{\texttt{nelson1@umd.edu}},
\href{mailto:jrajakum@ucla.edu}{\texttt{jrajakum@ucla.edu}},
\href{mailto:stan97@umd.edu}{\texttt{stan97@umd.edu}}}
\endgroup

\begin{abstract}
We give a deterministic classical algorithm that estimates $|\langle x|U|0^n\rangle|^2$ to additive error $\varepsilon$ in $\mathrm{poly}(n, 1/\varepsilon)$ time, where $U$ is a constant-depth quantum circuit comprised of gates with bounded fan-in and arbitrary connectivity, and $x$ is an arbitrary $n$-bit output string. This improves over prior state-of-the-art algorithms that takes $n^{O(log(n))}$ time for the same task, $n^{O(log(log(n))}$ when $U$ is geometrically local, and $n^{O(1)}$ for 2D geometrically-local circuits.
\end{abstract}

\section{Introduction}\label{sec:introduction}
Constant-depth quantum circuits provide a basic setting for studying the computational tasks that quantum computers can perform using only a few layers of gates applied in parallel. Such settings are currently experimentally relevant for quantum devices limited by short coherence times, and they are theoretically significant in investigations of separations between quantum and classical computational complexity. In particular, certain families of constant-depth quantum circuits are conjectured to generate distributions that cannot be sampled efficiently by classical computers. ~\cite{GWD,BGK,BV}, and these proposals have been recently implemented on experimental quantum devices~\cite{ringbauer2025verifiable,bluvstein2024logical}. Even without any complexity-theoretic conjectures, they are known to unconditionally outperform classical constant-depth circuits for certain relational and sampling problems~\cite{BGK,watts2019exponential,grier2025quantum,watts2026unconditional}, and these advantages continue to hold up in noisy settings~\cite{bravyi2020quantum,caha20263d}. A basic question is, \\

\textit{What sort of computational advantages over classical computers do constant-depth quantum circuits provide?}\\


Approaching this from the negative angle, improving the classical simulation of constant-depth quantum circuits helps delineate which computational advantages will fundamentally require high circuit depth or adaptive operations, and which can
already arise from a constant number of layers followed by measurement. 
We consider the following computational task, which is weaker than sampling: given the description of an $n$-qubit circuit implementing the unitary $U$, a
string $x\in\{0,1\}^n$, and $\varepsilon>0$, compute $\widehat P(x)$ with
\begin{equation}\label{eq:intro-target}
 |\widehat P(x)-P(x)|\le\varepsilon,
 \qquad P(x)=|\langle x|U|0^n\rangle|^2.
\end{equation}
One can obtain the estimate of $P(x)$ using a quantum computer by simply running the circuit repeatedly and counting the instances of $x$. It would be highly desirable from the perspective of quantum advantage if this task was difficult for classical computers: because it is defined as an observable estimation problem rather than a sampling problem, it more naturally lends itself to error mitigation and verification protocols~\cite{temme2017error,li2017efficient}.

The problem of additive-error output probability estimation for constant-depth circuits has a rich history starting with Bravyi, Gosset, and Movassagh in 2019~\cite{BGM} who gave a polynomial-time classical algorithm for circuits with nearest-neighbor gates in 2-dimensions. Successive works extended the geometric setting and eventually removed the geometric restriction, at the cost of quasipolynomial
time~\cite{BGM,CC,Dontha,BGL}. We obtain deterministic polynomial time
for every fixed depth and gate arity, without a restriction on connectivity.


\begin{theorem}[Main result]\label{thm:main}
Given a bitstring $x \in \{0,1\}^n$, an error tolerance $\varepsilon$, and the description of a constant-depth quantum circuit $U$ on $n$ qubits comprised of bounded fan-in gates with arbitrary connectivity,  a deterministic classical algorithm computes an estimate of $|\langle x|U|0^n\rangle|^2$ to additive error $\varepsilon$ in $\poly(n, 1/\varepsilon)$ time.
\end{theorem}

More generally, our classical algorithm can estimate probabilities for distributions over binary random variables defined on a bounded-degree strong dependency graph, given a
subroutine for evaluating probabilistic outcomes on small sets $S$ with runtime exponential in $|S|$.

At a high level, a common starting point for the latter three works that improve on the protocol by Bravyi, Gosset, and Movassagh is that estimating an $\varepsilon$-additive error output probability requires a nonzero estimate only when that probability exceeds $\varepsilon$. Therefore, in shallow quantum circuits, independence between regions of the circuit with disjoint light cones forces many regions to be strongly concentrated on the corresponding target bitstring. \cite{CC} and \cite{Dontha} exploit slices of the circuit with high probability weight to recursively decompose the circuit until they are small enough to approximate, while \cite{BGL} develops both a Hamming-ball approximation and a geometric dissection approach. These methods turn the combination of peakedness and restricted light-cone overlap into tractable classical approximations. 

The main improvement we make has to do with how the classical algorithm deals with the small set of regions that are not peaked. At a high-level, we introduce a cluster expansion technique to handle the regions of the output distribution that are peaked \cite{mann2025approximate,Barvinok}, and use an inclusion-exclusion argument to bound the approximation error that correlations with the non-peaked regions could introduce. The cluster expansion framework has recently gained tremendous interest in the setting of quantum information ~\cite{
harrow_classical_algorithms_correlation,
mann_efficient_algorithms_approximating,
mann_minko_2024,
chen_rouze_chen_fermions_2025,bakshi_high-temperature,yin_polynomial-time_classical_sampling,zlokapa_quasipolynomial_syk_2026,zlokapa_syk,midha2026beyond,midha2026algorithmic,CE_shorttimeLRB23,CE_metastab25,CE_ldpcgap25}, as it provides a way to study systems that are perturbed from exactly solvable points \cite{CE_Mayer1937,kotecky_preiss_1986,
netocny_redig_2004,friedli_velenik}. In the setting of output probability estimation for quantum circuits, cluster expansion (or similar techniques) have been applied before \cite{BGM,harrow_classical_algorithms_correlation,mann_minko_2024}, but most prior work uses the identity or maximally mixed state as the exactly solvable point. In our work, we instead choose the peaked computational basis state as our exactly solvable point. In this formulation, we are able to directly apply a classical algorithm of \cite{mann2025approximate} as a black box that computes sufficiently peaked marginals of the output distribution. Our main technical contribution is a reduction to this subroutine via inclusion-exclusion techniques.


\begin{table}[t]
\centering
\small
\renewcommand{\arraystretch}{1.12}
\begin{tabularx}{\textwidth}{@{}l X l@{}}
\toprule
Work & Connectivity & Running time \\
\midrule
Bravyi--Gosset--Movassagh~\cite{BGM}
 & Arbitrary & $2^{\widetilde O(\sqrt n)}$ \\
Bravyi--Gosset--Movassagh~\cite{BGM}
 & Two-dimensional grid & $\poly(n)$ \\
Coble--Coudron~\cite{CC}
 & Three-dimensional grid & $n^{(\log n)^{O(1)}}$ \\
Dontha et al.~\cite{Dontha}
 & Grid in fixed dimension & $n^{(\log n)^{O(1)}}$ \\
Bravyi--Gosset--Liu~\cite{BGL}
 & Arbitrary & $n^{O(\log n)}$ \\
Bravyi--Gosset--Liu~\cite{BGL}
 & Grid in fixed dimension $D>2$ & $n^{O(\log\log n)}$ \\
This work & Arbitrary & $\poly(n)$ \\
\bottomrule
\end{tabularx}
\caption{Comparison of running times for output probability estimation for constant-depth circuits with gates of bounded fan-in up to
additive error that is inverse polynomial in $n$.}
\label{tab:intro-comparison}
\end{table}

The comparison in Table~\ref{tab:intro-comparison} concerns estimation of a specified output probability. \cite{BGL} additionally give a quasipolynomial-time approximate-sampling algorithm and, in their general construction, an approximation to the output state, under the peakedness assumption $\max_x P(x)\ge n^{-a}$ for a fixed constant $a$. We emphasize that
Theorem~\ref{thm:main} gives polynomial-time probability
estimation without any peakedness assumption, but does not
provide polynomial-time algorithms for those stronger tasks.
In particular, separate additive estimates of individual probabilities need not form a normalized distribution or guarantee bounded total variation error.
Extending our work to polynomial-time sampling
of arbitrary-connectivity peaked shallow circuits is a natural open question.


\subsection{Proof Strategy}

Suppose the target output bitstring is $x \in \{0,1\}^n$. Since this is a shallow circuit, we can brute-force compute the probability that any single output bit agrees with the corresponding bit in $x$. Consider the set of \textit{bad} output bits $\mathcal B$ whose probability of agreeing with their corresponding bit in $x$ is less than some threshold $1-\tau$. The starting point of \cite{BGL} (and prior work on this matter) is to observe that in shallow quantum circuits, the probability of measuring the outcome $x$ is bounded by $(1-\tau)^{\Theta(|\mathcal B|)}$, and so $|\mathcal B|$ must be small if $x$ is to be measured with non-negligible probability. On the remaining set of \textit{good} output bits, the output distribution is heavily concentrated around bitstrings that mostly agree with $x$, with deviations bounded by $\tau$. It is known how to estimate marginals of such distributions to relative-error, via a cluster expansion that converges when $\tau$ is a sufficiently low constant~\cite{mann2025approximate,Barvinok}. Our main contribution is to show that the output probability of $x$ can be expressed as a sum of marginals on the \textit{good} bits, weighted by marginals on regions that connect the \textit{good} and \textit{bad} bits. In particular, we show that enumerating and evaluating this expression using only small connecting regions produces the desired additive error estimate.

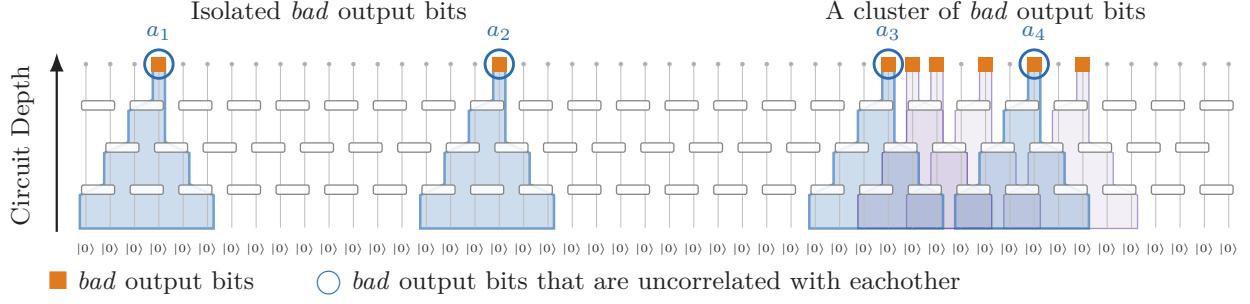
\begin{figure}[h]
\centering
\def\lightcone#1#2{%
  \pgfmathtruncatemacro{\wparity}{mod(#1,2)}%
  \ifnum\wparity=1
    \pgfmathtruncatemacro{\La}{#1}\pgfmathtruncatemacro{\Ra}{#1+1}%
  \else
    \pgfmathtruncatemacro{\La}{#1-1}\pgfmathtruncatemacro{\Ra}{#1}%
  \fi
  \pgfmathtruncatemacro{\Lb}{\La-1}\pgfmathtruncatemacro{\Rb}{\Ra+1}%
  \pgfmathtruncatemacro{\Lc}{\Lb-1}\pgfmathtruncatemacro{\Rc}{\Rb+1}%
  \path[#2]
    (#1-.25,4.05) -- (#1-.25,3.23) -- (\La-.25,2.97) -- (\La-.25,2.23)
    -- (\Lb-.25,1.97) -- (\Lb-.25,1.23) -- (\Lc-.25,.97) -- (\Lc-.25,.18)
    -- (\Rc+.25,.18) -- (\Rc+.25,.97) -- (\Rb+.25,1.23) -- (\Rb+.25,1.97)
    -- (\Ra+.25,2.23) -- (\Ra+.25,2.97) -- (#1+.25,3.23) -- (#1+.25,4.05) -- cycle;
}
\resizebox{\linewidth}{!}{%
\begin{tikzpicture}[x=.30cm,y=.52cm,font=\small,text=figInk,>=Latex,
  cone/.style={fill=figAttach,fill opacity=.10,draw=figAttach,draw opacity=.65,line width=.45pt},
  selcone/.style={fill=figBlue,fill opacity=.23,draw=figBlue,draw opacity=.65,line width=.9pt},
  gate/.style={rounded corners=.7pt,fill=white,fill opacity=.87,draw=figInk!55,line width=.45pt},
  bad/.style={rectangle,fill=figBad,minimum size=5pt,inner sep=0pt}]
\node[anchor=west,font=\bfseries] at (-1,5.9) {};
\node[font=\footnotesize] at (11,5.3) {Isolated \textit{bad} output bits};
\node[font=\footnotesize] at (38,5.3) {A cluster of \textit{bad} output bits};
\draw[->,figInk,line width=.8pt] (-.2,.1)--(-.2,4.35);
\node[rotate=90,font=\footnotesize] at (-1.6,2.15) {Circuit Depth};

\foreach \w in {35,36,38,42} {\lightcone{\w}{cone}}
\foreach \w in {4,18,34,40}  {\lightcone{\w}{selcone}}

\foreach \w in {1,...,48} {
  \draw[figInk!30,line width=.35pt] (\w,.2)--(\w,4.05);
  \node[font=\fontsize{4}{5},anchor=north] at (\w,.08) {$|0\rangle$};
}
\foreach \l in {1,3} {
  \foreach \w in {1,3,...,47} {
    \draw[gate] (\w-.18,\l-.01) rectangle (\w+1.18,\l+.21);
  }
}
\foreach \w in {2,4,...,46} {
  \draw[gate] (\w-.18,1.99) rectangle (\w+1.18,2.21);
}
\foreach \w in {1,...,48} {\fill[figInk!35] (\w,4.08) circle (.8pt);}
\foreach \w in {4,18,34,35,36,38,40,42} {\node[bad] at (\w,4.08) {};}
\foreach \w/\k in {4/1,18/2,34/3,40/4} {
  \draw[figBlue,line width=1pt] (\w,4.08) circle (5pt);
  \node[text=figBlue,font=\scriptsize,above] at (\w,4.4) {$a_\k$};
}

\node[anchor=west,font=\footnotesize] at (-1,-1.1)
 {\textcolor{figBad}{$\blacksquare$} \textit{bad} output bits
 \qquad \textcolor{figBlue}{$\bigcirc$} \textit{bad} output bits that are uncorrelated with eachother};
\end{tikzpicture}%
}

\caption{Depiction of `lightcones' in a 1D quantum circuit. The \textit{bad} bits of the output distribution agree with $x$ with probability $\leq 1-\tau$. The separation of the blue lightcones illustrates how the probability of $a_1,a_2,a_3,a_4$ all agreeing with $x$ is $\leq (1-\tau)^4$.}
\label{fig:lightcones-bad-bits}
\end{figure}

To provide intuition for this decomposition, we first construct a connectivity graph $G$ by drawing edges between bits that could be correlated with each other.  In shallow circuits, this graph has bounded degree, because two output bits need to share a ``light cone'' in the circuit in order to be correlated. We will apply an inclusion-exclusion argument and use these lightcone properties to write $|\langle x|U|0^n\rangle|^2$ in the following form:
\begin{align}
    &\Pr[\text{all bits agree with $x$}] \notag \\
    &=\notag
\sum_{\substack{T \subseteq [n] \backslash \mathcal B\text{: all connected components of $G[T]$ touch $\mathcal B$}}}
(-1)^{|T|} \\
&\times \Pr[\text{all bits in $\mathcal{B}$ agree with $x$, \textcolor{figAttach}{all bits in $T$ disagree with $x$}}] \label{eq:bad}\\
&\times \Pr[\text{all bits outside the neighborhood of $T \cup \mathcal B$ agree with $x$}] \label{eq:good}
\end{align}

 Line (\ref{eq:bad}) isolates how the positive influence of the \textit{good} bits extends into the \textit{bad} bits: because \textit{good} bits should be unlikely to disagree with $x$ (as we have highlighted in purple), this term decays roughly exponentially with the number of such good bits, i.e. $|T|$. For sufficiently small $\tau$, this decay dominates the number of possible choices of $T $. Thus, we can enumerate only the terms corresponding to small-sized $T $, truncate the remaining terms, and obtain the desired additive error. In the resulting approximation, the term in line (\ref{eq:bad}) can be computed by brute force because the aforementioned argument of Bravyi-Gosset-Liu \cite{BGL} guarantees that $\mathcal B$ is also small (so the full set $T \cup \mathcal B$ is small). The term in line (\ref{eq:good}) can be computed using the aforementioned cluster expansion \cite{mann2025approximate,Barvinok} because it is a marginal on only \textit{good} bits.

 \input{good-bad-decomposition}

\subsection{Implications for expectation values}
\label{sec:intro-expectations}
Theorem~\ref{thm:main} also gives polynomial-time additive
estimation of the magnitude of product-operator expectation
values. For single-qubit operators $A_i$ with
$\|A_i\|\le1$, let
\[
 m=\langle0^n|U^\dagger
       \left(\bigotimes_{i=1}^n A_i\right)
       U|0^n\rangle.
\]
The construction used in~\cite{BGM,BGL} expresses $|m|^2$ as an output
probability of a circuit on at most $2n$ qubits with depth
at most $2d+1$, where $d$ is the depth of $U$.
Estimating this probability to additive error
$\varepsilon^2$, projecting the estimate onto $[0,1]$,
and taking its square root therefore estimates $|m|$
to additive error $\varepsilon$ in
$\poly(n,1/\varepsilon)$ time for fixed depth and gate arity.

When $m$ is known to be nonnegative, this procedure also
estimates $m$ itself. For a general real expectation value,
however, the squared magnitude does not determine its sign. Similarly,
for a complex expectation value, it does not determine its
phase. Thus the task of polynomial-time estimation of general signed or complex expectation values remains an open problem.

\section{Setup and preliminaries}\label{sec:model}
Throughout, $\log$ denotes the natural logarithm and $\log_2$ the base-two logarithm. Denote the set $\{1, \ldots, n\}$ for $n \in \mathbb{N}$ as $[n]$. Our computational model includes all $n$-qubit, depth-$d$, and $k$-local quantum circuits where the gates at each of the $d$ layers of the circuit act on disjoint sets of at most $k$ qubits. Following the brute-force-style calculations in~\cite{BGM}, we state running times for computing output probabilities exactly. Thus, the approximation error comes from truncating inclusion-exclusion and the subroutine $\mathsf{Good}$ defined in Section~\ref{sec:algorithm}.

Following the notation in $\cite{BGL}$, we are interested in estimating the output probability
\[P(x) = |\langle x | U |0^n \rangle|^2,\]
where $U$ is an $n$-qubit, depth-$d$, and $k$-local quantum circuit and $x$ is a length-$n$ variable bitstring.
For the purpose of our work, it suffices to construct a classical algorithm to estimate $P(0^n) = |\langle 0^n | U | 0^n \rangle|^2$.
For any computational basis input or output strings $x, y$,  we can simply conjugate $U$ with the appropriate bit-flip operators to give us $|\langle y | U | x \rangle|^2 = |\langle 0^n | X^y U X^x | 0^n \rangle|^2$. For $S \subseteq [n]$, define $x_S$ to be the restriction of $x$ to the indices in $S$. Further define the marginal $P_S(x_S)$ to be the probability of measuring $x_S$ on all qubits in $S$, and abbreviate to $P(x_S)$ when it is clear which qubits are involved.
Throughout, we will use the notation $0_S$ and $1_S$ to denote the measurement outcome bitstring of all $0$'s or all $1$'s, respectively, on the qubits in $S$.
When $S$ is the singleton set $\{i\}$, we abbreviate the measurement outcome as $0_i$ and $1_i$.
Finally, let $\widehat{P}$ denote the probability estimator and use the same notation convention for the marginal probabilities.
For a specified outcome $x_S$, let $\widehat P(x_S)$ denote an estimate of $P(x_S)$. Note that estimates for different outcomes or subsets need not be marginals of a common probability distribution.


Define the backward light cone of an output qubit $i$ as the set of input qubits with a directed path through the gates in the circuit to output qubit $i$. Let $\mathcal{L}(i)$ be this set. Similarly, we abuse the notation slightly and define
\[\mathcal{L}(A) \equiv \bigcup_{j \in A} \mathcal{L}(j),\qquad A \subseteq [n].\]
If two subsets of qubits $A, B \subseteq [n]$ satisfy the property that $\mathcal{L}(A) \cap \mathcal{L}(B) = \emptyset$, then we say that $A$ and $B$ are \emph{lightcone-separated}.
We construct a \emph{light-cone overlap graph} $G = (V, E)$ where $|V| = n$. We assign one vertex per output qubit and introduce an edge $(i, j) \in E$ if and only if $\mathcal{L}(i) \cap \mathcal{L}_(j) \neq \emptyset$ for $i \neq j$. Denote the maximum degree of $G$ as $K$. A simple light-cone bound gives us $K \leq k^{2d}$. We denote closed neighborhoods of a subset of vertices $S \subseteq V$ as $N_G[S]$ where $S \subseteq N_G[S]$.

We present a set of well-known light-cone properties in the following lemma.

\begin{lemma}[Basic light-cone facts]\label{lemma:light-cone-facts}
    The graph $G$ can be constructed in polynomial time for fixed constants $d$ and $k$. The marginal has the following properties:
    \begin{enumerate}
        \item $0\le P(x_S)\le P(x_T)\le1$ whenever $T\subseteq S$.
        \item If no edge joins disjoint sets $S,T$, their full measurement output vectors are independent. In particular,
        \[
         P(x_{S\cup T})=P(x_S)P(x_T).
        \]
    \end{enumerate}
\end{lemma}
\begin{proof}
    Because each layer of the backward light cone increases the number of relevant qubits by at most a factor of $k$, the size of a backward light cone for any output qubit $i$ is given by $|\mathcal{L}(i)| \leq k^d$. Similarly, each of the input qubits in $\mathcal{L}(i)$ can reach a set of at most $k^d$ output qubits using the same argument. Therefore, every output qubit has a backward light cone that has nontrivial overlap with the backward light cone at most $k^{2d}$ other output qubits, giving us $K \leq k^{2d}$ which is a constant. Enumerating these bounded-size cones takes constant time to construct the edges for each of the $n$ vertices so the graph can be constructed in $O_{d, k}(nk^{2d})$ operations with polynomial bit overhread for labels.
    For (1), the monotocity shown in the expression is immediate from the event interpretations of $P(x_S)$ and $P(x_T)$.
    For (2), the absence of an edge joining sets $S$ and $T$ implies that the backpropagated measurement operators on $S$ and $T$ necessarily act on disjoint input qubits. Since the input qubits are assumed to be in a product state, the joint expectation with respect to the measurement operators will factor. This proves the independence of the two entire collections.
\end{proof}

\section{Isolating the bad bits}\label{sec:core}
First, we will choose some threshold $0\leq\tau\leq 1 $ and define $\mathcal{B}$ and $\mathcal{G}$ as the set of bad and good bits respectively,
\[\mathcal{B} = \left\{i\,:\,P(1_i) > \tau \right\},\qquad \mathcal{G} = [n] \setminus \mathcal{B}\]
For reasons that will be apparent in Section \ref{sec:small}, we will fix 

\begin{align}
    \tau =
 \left(\frac{1}{e^{2}(2K+1)}\right)^{K+1}
\end{align}

First, we have the following bound on the size of $\mathcal{B}$

\begin{lemma}\label{lemma:when_to_output_zero}
    If $\sum_{i \in [n]} P(1_{i}) \geq (K+1)\log\frac{1}{\varepsilon}$, then $P(0^n) \leq\varepsilon$. Otherwise,
    \begin{align}
        |\mathcal B| \leq (K+1)\left(\frac{1}{e^{2}(2K+1)}\right)^{-(K+1)} \log\frac{1}{\varepsilon}
    \end{align}
\end{lemma}
\begin{proof}
Color the vertices of the graph $G$ with $K + 1$ colors. By the pigeonhole principle, we can find at least one color $C$ such that for all output qubits $i \in C$, we have $\sum_{i \in C} P(1_i) \geq  \sum_{i \in [n]} P(1_{i}) / (K + 1)$. The output qubits in $C$ are mutually independent, i.e., they are light-cone separated, and therefore,
    \begin{align*}
        P(0^n) &\leq \prod_{i \in C} (1 - P(1_i)) \leq e^{-\sum_{i \in C}P(1_i)} \leq e^{-\sum_{i \in [n]} P(1_i)/(K+1)}
    \end{align*}
    Thus, when $\sum_{i \in [n]} P(1_{i}) \geq (K+1)\log\frac{1}{\varepsilon}$, we have the preliminary condition. Otherwise, observe that each bit in $\mathcal{B}$ contributes $\tau$ to $\sum_{i \in [n]} P(1_{i})$, i.e. 
    \begin{align}
        |\mathcal{B}|\tau \leq \sum_{i \in [n]} P(1_{i}) \leq (K+1)\log \frac{1}{\varepsilon}
    \end{align}
    which obtains the main bound.
\end{proof}

\subsection{Separating the contributions of good and bad bits}
The core insight of our work is to pull apart the joint probability $P(0^n)$ into marginal probabilities of good bits and marginal probabilities of bad bits with a limited number of attached good bits.

An inclusion-exclusion decomposition expresses $P(0^n)$ in terms of mixed outcomes that are zero on the bad bits and one on selected good bits:
\begin{align}
    P(0^n) =\sum_{S \subseteq [n] \backslash B }(-1)^{|S|}P(0_{\mathcal B},1_S)
\end{align}
Our starting point is to observe that many of the terms in the sum decompose into a product of probabilities due to the shallowness of the circuit. Recall from Section~\ref{sec:model} that we can construct a light-cone overlap graph $G$ that has constant degree by assigning edges between vertices that correspond to bits that are correlated with each other. Now, we will write each $S$ as $S = T \sqcup S \backslash T$ where $T$ contains all the bits of $S$ that form connected components adjacenet to $\mathcal B$ . Further, for any $T$, let $T' =[n]\setminus N_G[\mathcal B\cup T] \subseteq [n]\backslash \mathcal B$ denote all bits that are not adjacent to $T$ or $\mathcal B$.
\begin{align}
    P(0^n) &=    \sum_{S \subseteq [n] \backslash \mathcal B }(-1)^{|S|}P(0_{\mathcal B},1_T) P(1_{S\setminus T})\\
    &=    \sum_{T \text{ consists of conn. comp. touching } \mathcal B}(-1)^{|T|}P(0_{\mathcal B},1_T)\left[\sum_{S \subseteq T'}(-1)^{|S|}P(1_S)\right]\\
    &=    \sum_{T \text{ consists of conn. comp. touching } \mathcal B}(-1)^{|T|}P(0_{\mathcal B},1_T)\left[P(0_{T'}) \right]\\
\end{align}

However, this is still not a tractable decomposition to evaluate since $T$ can be very large. We next prove that we can safely truncate the sum to only include small sets $T$.

\subsection{Bounding the bad bits with few good bits connected to them}
Notice that by definition, good bits have a small marginal probability of disagreeing with target string $0^n$. In fact, the joint probability decays exponentially as the number of good bits involved increases. This means that the terms associated with sets $T$ that have many good bits have very small contribution, and can be safely truncated as we show next.

\begin{lemma}\label{lem:truncating_error}
    \begin{equation}
        |P(0^n)-\widehat{P}(0^n)| \leq \varepsilon ,
    \end{equation}
    where
    \begin{equation}
        \widehat{P}(0^n) =\sum_{\substack{T \text{ path-connected to } \mathcal B\\|T| \leq \ell}}(-1)^{|T|}P(0_{\mathcal B},1_T)\left[P(0_{T'}) \right]\\ \label{eq:good-bad-decomp}
    \end{equation}
    and
    \begin{equation}
        \ell = |\mathcal B|  + \left\lceil\log_2 \frac{1}{\varepsilon}\right\rceil
    \end{equation}
\end{lemma}

\begin{proof}
    Our goal is to bound, 
    \begin{align}
        |P(0^n)-\widehat{P}(0^n)| &= \left|
\sum_{\substack{T \text{ consists of conn. comp. touching } \mathcal B\\ |T| > \ell}}
(-1)^{|T|}
P(0_{\mathcal B},1_T)
\left[P(0_{T'})\right]
\right| \\
& \leq \sum_{\substack{T \text{ consists of conn. comp. touching } \mathcal B\\ |T| > \ell}}
P(0_{\mathcal B},1_T)
    \end{align}
    Consider any $T$ of size $|T| = s$. Note that by choosing $s/(K+1)$ qubits in $T$ with non-overlapping lightcones, we have
    \begin{align}
        P(0_{\mathcal B},1_T) \leq \tau^{s/(K+1)}
    \end{align}
    Note the exponential decay in $s$. To finish the bound, we must bound the number of valid choices of $T$ of size $s$. Fix an ordering $b_1,\ldots,b_{|\mathcal B|}$ of the vertices in $\mathcal B$. For each valid $T$, assign each connected component of $G[T]$ to the first vertex in this ordering adjacent to that component. Such a vertex exists because $T$ is path-connected to $\mathcal B$.

Let $s_i$ be the total number of vertices in the components assigned to $b_i$. Then $s_1+\cdots+s_{|\mathcal B|}=s$, with each $s_i\geq 0$. For $|\mathcal B|\geq 1$, there are
$\binom{s+|\mathcal B|-1}{|\mathcal B|-1}\leq 2^{|\mathcal B|+s}$
possible choices of these sizes. The union of $b_i$ and all components assigned to it is a connected set of size $s_i+1$ containing $b_i$, so there are at most $(eK)^{s_i}$ choices for this union. These unions determine $T$, and hence
\begin{align}
    \left|\{T:T \text{ path-connected to } \mathcal B,\ |T|=s\}\right|
    \leq 2^{|\mathcal B|+s}\prod_{i=1}^{|\mathcal B|}(eK)^{s_i}
    = 2^{|\mathcal B|+s}(eK)^s.
\end{align}

    
    Putting this together:
    \begin{align}
    \sum_{\substack{T \text{ path-connected to } \mathcal B\\ |T| > \ell}}
    P(0_{\mathcal B},1_T)
    &\leq \sum_{s=\ell+1}^{\infty}
    \sum_{\substack{T \text{ path-connected to } \mathcal B\\ |T| = s}}
    P(0_{\mathcal B},1_T) \\
    &\leq 2^{|\mathcal B|} \sum_{s=\ell+1}^{\infty}
    \left(2eK\tau^{1/(K+1)}\right)^s \\
    &= 2^{|\mathcal B|} \sum_{s=\ell+1}^{\infty}
    \left(\frac{2K}{e(2K+1)}\right)^s \\
    &< 2^{|\mathcal B|} \sum_{s=\ell+1}^{\infty} \frac{1}{2^s} \\
    &= 2^{|\mathcal B|-\ell} \\
    &\leq \varepsilon,
\end{align}
    where in the second equality we plugged in $\tau = \left(\frac{1}{e^2(2K+1)}\right)^{K+1}$ and in the last inequality we used $\ell = |\mathcal B|+\left\lceil\log_2\frac{1}{\varepsilon}\right\rceil$.
\end{proof}

\subsection{Computing the marginal probability of the good bits using cluster expansion}\label{sec:small}
It remains to compute the marginals $P(0_W)$ for $W \subseteq \mathcal{G}$. Since each qubit in $\mathcal{G}$ is ``good'', they each have a low marginal probability of disagreeing with the target bit. Furthermore, by the light-cone structure of the circuit, each output bit only depends on a limited number of other output bits. It turns out that joint probabilities with this structure can be efficiently  estimated to small relative error using cluster expansion \cite{mann2025approximate,Barvinok}. In fact, we can apply the following theorem directly by letting $E_v$ be the event that the output of qubit $v$ agrees with the target bit.
\begin{theorem}[Mann--Waite, Theorem 2 of~\cite{mann2025approximate}]\label{thm:MW}\label{thm:MM}
Fix $K \in\mathbb{Z}_{\ge2}$. Let $\{E_v\}_{v\in V(G)}$ be a set of events in a probability space with strong dependency graph $G$ of maximum degree $K$. Suppose that, for all $v\in V(G)$,
\[
 \Pr\!\left[\overline{E_v}\right]\le
 \left(\frac{1}{e^{2}(2K+1)}\right)^{K+1}.
\]
If for all $S\subseteq V(G)$, $\Pr\left[\bigcap_{v\in S}\overline{E_v}\right]$ can be computed in time $\exp(O(|S|))$, then there is a fully polynomial-time approximation scheme for estimating probability of $\Pr\left[\bigcap_{v\in V(G)}E_v \right]$ that has relative error $\varepsilon$ and runtime $\text{poly}(|V(G)|,\varepsilon^{-1})$.
\end{theorem}

 \begin{corollary} \label{corr:MW}
     When $K \geq 2$, for any set of good bits, $T' \subseteq \mathcal{G}$, there exists a deterministic classical algorithm that estimates $P(0_{T'}) $ to multiplicative error $\varepsilon$ in runtime $\poly(|T'|,\varepsilon^{-1})$.
 \end{corollary}
 \begin{proof}
     For our all-zero target, $\Pr[\overline{E_v}]=P(1_v)$, $\Pr[\bigcap_{v\in S}\overline{E_v}]=P(1_S)$, $\Pr[\bigcap_{v\in W}E_v]=P(0_W)$, and $P(1_v) \le \tau =  \left(\frac{1}{e^{2}(2K+1)}\right)^{K+1}$. 
 \end{proof}

We can handle the $K=0$ and $K=1$ cases trivially.

\section{Algorithm design}\label{sec:algorithm}

We combine the inclusion-exclusion decomposition detailed in Section~\ref{sec:core} with the
Mann--Waite result of Section~\ref{sec:small} for our classical algorithm which we state in the following theorem.

\begin{theorem}[Main theorem]
\label{thm:algorithm-main}
Fix positive integers $d,k$. Let $U$ be an $n$-qubit unitary
quantum circuit with $d$ layers, each consisting of gates on disjoint sets
of at most $k$ qubits. Let $L_{\rm in}$ be the length of its classical
description.
Given $x\in\{0,1\}^n$ and $0<\varepsilon<1/2$, there is a
deterministic classical algorithm using $\poly_{d,k}(n,L_{\rm in},1/\varepsilon)$ operations returning
$\widehat P(x)\in[0,1]$ such that
\[
 \left|\widehat P(x)-|\langle x|U|0^n\rangle|^2\right|
 \le\varepsilon.
\]
\end{theorem}

\subsection{Definitions and subroutines}
We reduce an arbitrary target string to $0^n$ by appending the appropriate single-qubit bit flips. This adds at most one layer and does not enlarge the backward light cones, so their size bound still uses the original depth $d$. We use the following two subroutines. The subroutine $\mathsf{Brute}(S,z_S)$ returns $P(z_S)$ exactly. After a polynomial-time construction of the backward light cone corresponding to the set $S$ of output bits with $|S| \geq 1$, we can evaluate $P(z_S)$ exactly with time
\begin{equation}\label{eq:brute-runtime}
 T_{\mathsf{Brute}}(S)
 =O\!\left(dk^d |S|\,2^{k^d |S|+k}\right)
 =2^{O_{d,k}(|S|)}.
\end{equation}
The extra bit flips we append do not change the asymptotic bound. This is a $k$-local generalization of \cite{BGM}'s exact calculation in time $2^{O(2^d s)}$ for two-qubit gates. We emphasize that $\mathsf{Brute}$ returns the exact probability value.

For $W\subseteq\mathcal G$ and $0<\eta<1$, the subroutine $\mathsf{Good}(W,\eta)$ returns an estimate in $[0,1]$ of $P(0_W)$ with additive error at most $\eta$, in $\poly_{d,k}(n,L_{\rm in},1/\eta)$ operations. The subroutine essentially applies Theorem~\ref{thm:MW} with relative error $\eta$ and evaluates its marginals $P(1_S)$ exactly by $\mathsf{Brute}$. Since $P(0_W)\le 1$, a relative error of at most $\eta$ implies an additive error of at most $\eta$. Note that projecting the resulting estimate onto $[0,1]$ cannot increase its error. Both subroutines are also defined to return one on the empty set.

\subsection{Algorithm}
We construct Algorithm~\ref{alg:main} for the classical algorithm stated in Theorem~\ref{thm:algorithm-main}.
\begin{algorithm}[H]
\caption{Additive estimation of a specified output probability}
\label{alg:main}
\small
\begin{algorithmic}[1]
\Require A circuit $U$, a target string $x\in\{0,1\}^n$,
         and $0<\varepsilon<1/2$.
\Ensure A $\widehat P(x)\in[0,1]$ satisfying
        $|\widehat P(x)-P(x)|\le\varepsilon$.

\State $U \gets \left(\prod_{i=1}^n X_i^{x_i}\right)U$.
\State Construct the light-cone overlap graph $G$ and its maximum degree $K$ for $U$.

\If{$K\le1$}
    \State Let $C_1,\ldots,C_c$ be the connected components of $G$.
    \For{$j=1,\ldots,c$}
        \State $a_j\gets
        \mathsf{Brute}(C_j,0_{C_j})$.
    \EndFor
    \State \Return $\min\{1, \max\{0, \prod_{j=1}^c a_j\}\}$.
\EndIf

\State Set $\tau\gets[e^2(2K+1)]^{-(K+1)}$.
\For{$i=1,\ldots,n$}
    \State $P(1_i)\gets\mathsf{Brute}(\{i\},1_i)$.
\EndFor
\If{$\sum_{i=1}^n P(1_i) \ge(K+1)\log(1/\varepsilon)$}
    \State \Return $0$.
\EndIf

\State $\mathcal{B}\gets\{i: P(1_i) > \tau\}$;
       $\mathcal{G}\gets[n]\setminus \mathcal{B}$.
\If{$\mathcal{B}=\emptyset$}
    \State \Return $\min\{1, \max\{0, \mathsf{Good}([n],\varepsilon)\}\}$.
\ElsIf{$\mathcal{G}=\emptyset$}
    \State \Return $\min\{1, \max\{0, \mathsf{Brute}([n],0^n)\}\}$.
\EndIf
\State $\ell\gets |\mathcal{B}| +\lceil\log_2(2/\varepsilon)\rceil$.
\State Enumerate without repetitions
       $\mathcal A\gets
       \{T\subseteq \mathcal{G}:T\text{ is path-connected to }\mathcal{B},\ |T|\le\ell\}$.
\Statex \hspace{\algorithmicindent}
        Include the empty attachment $T=\varnothing$.
\State $a\gets0$.

\ForAll{$T\in\mathcal A$}
    \State $T'\gets[n]\setminus N_G[\mathcal{B}\cup T]$.
    \State $u\gets\mathsf{Brute}(\mathcal{B}\cup T,(0_\mathcal{B},1_T))$.
    \State $v\gets\mathsf{Good}(T',\varepsilon/[2(|\mathcal{A}|+1)])$.
    \State $a\gets a+(-1)^{|T|}uv$
\EndFor

\State \Return $\min\{1, \max\{0, a\}\}$
\end{algorithmic}
\end{algorithm}

\subsection{Correctness}
Every call to $\mathsf{Brute}$ returns an exact probability and all arithmetic operations are assumed to be exact as per prior work~\cite{BGM,BGL}.
Thus, there are two possible sources of approximation error: truncating the inclusion-exclusion expansion and approximating the remaining marginals of the good bits with $\mathsf{Good}$.

\begin{proof}[Proof of correctness in
Theorem~\ref{thm:algorithm-main}]
After applying the output bit flips, $P(0^n)$ for the modified circuit
is the requested probability $P(x)$ for the original circuit. We
therefore prove the additive-error guarantee for the modified circuit.

If $K\le1$, each connected component $C_j$ has at most two vertices,
and the output vectors on distinct components are jointly independent
by Lemma~\ref{lemma:light-cone-facts}. Every factor $a_j=P(0_{C_j})$ is exact, so their product equals $P(0^n)$ and this branch has zero error.

Now suppose $K\ge2$. The singleton probabilities and their sum
$\sum_iP(1_i)$ are exact. If
$\sum_iP(1_i) \ge(K+1)\log(1/\varepsilon)$, Lemma~\ref{lemma:when_to_output_zero} gives $P(0^n)\le\varepsilon$, so returning zero has error at most $\varepsilon$.

Otherwise, $\sum_iP(1_i)<(K+1)\log(1/\varepsilon)$. The exact definition
$\mathcal B=\{i:P(1_i)>\tau\}$ ensures that every $i\in\mathcal G$ satisfies $P(1_i)\le\tau$. For any $W\subseteq\mathcal G$, Lemma~\ref{lemma:light-cone-facts} makes $G[W]$ a strong dependency graph of maximum degree at most $K$. Its marginals can be computed exactly in $2^{O_{d,k}(|S|)}$ operations by $\mathsf{Brute}$ so all hypotheses of Theorem~\ref{thm:MW} hold.
If $\mathcal B=\emptyset$, the returned value
$\mathsf{Good}([n],\varepsilon)$ estimates $P(0^n)$ directly with error at most $\varepsilon$.
If $\mathcal G=\emptyset$, the returned value
$\mathsf{Brute}([n],0^n)$ equals $P(0^n)$ exactly and this branch has zero error.

It remains to consider the branch where both sets are nonempty.
For each $T$, write
$T'=[n]\setminus N_G[\mathcal B\cup T]$. The exact decomposition is
\begin{equation}\label{eq:correctness-decomposition}
 P(0^n)=
 \sum_{\substack{T\subseteq\mathcal G\\
                  T\text{ path-connected to }\mathcal B}}
 (-1)^{|T|}
 P(0_{\mathcal B},1_T)P(0_{T'}).
\end{equation}

To verify this identity, apply inclusion--exclusion to the good
output bits while enforcing zero on $\mathcal B$. For a selected set
$S\subseteq\mathcal G$, let $T$ contain precisely the components
of $G[S]$ that have a neighbor in $\mathcal B$. Then $S\setminus T$
has no edge to $\mathcal B\cup T$, and the corresponding output
vectors are independent. For each fixed $T$, the possible remainders
are exactly all subsets of $T'$. Their inclusion--exclusion sum is
$P(0_{T'})$, proving the identity.

Let $\mathcal A$ contain all possible $T$s of size at most
$\ell=|\mathcal B|+\lceil\log_2(2/\varepsilon)\rceil$, including the empty set after performing the truncation. 
Lemma~\ref{lem:truncating_error} gives \[\left|P(0^n)-\sum_{T\in\mathcal A}(-1)^{|T|}P(0_{\mathcal B},1_T) P(0_{T'})\right|\le\varepsilon/2.\]
For each retained $T \in \mathcal A$, the algorithm computes $P(0_{\mathcal B},1_T)$ exactly and
$\mathsf{Good}(T',\eta)$ with $\eta=\varepsilon/[2(|\mathcal A|+1)]$.
Since $0\le P(0_{\mathcal B},1_T)\le1$,
\[
 |P(0_{\mathcal B},1_T)\mathsf{Good}(T',\eta) -P(0_{\mathcal B},1_T)P(0_{T'})|=P(0_{\mathcal B},1_T)|\mathsf{Good}(T',\eta) -P(0_{T'})|\le\eta.
\]
The accumulated sum $a$ therefore satisfies
\[
 \left|a - \sum_{T\in\mathcal A}(-1)^{|T|}P(0_{\mathcal B},1_T) P(0_{T'})\right|\le |\mathcal A|\eta
 =\frac{|\mathcal A|\varepsilon}{2(|\mathcal A|+1)}<\frac\varepsilon2,
 \qquad
 |a-P(0^n)|<\frac\varepsilon2+\frac\varepsilon2
 =\varepsilon.
\]
Projection onto $[0,1]$ cannot increase the error. We now see that every branch has additive approximation error at most $\varepsilon$.
\end{proof}

\subsection{Time complexity}
Next, we show that the algorithm has polynomial runtime while deferring the more explicit calculation to the appendix~\ref{app:time}.
\begin{proof}[Proof of the time complexity statement in
Theorem~\ref{thm:algorithm-main}]

The dominant contribution to the runtime comes from enumerating over sets $T$ that are path-connected to $\mathcal B$ and, for each such set, computing both the contribution from $\mathcal B \cup T$ by brute force and the contribution from $T'$ using cluster expansion. To bound this runtime, first note that the number of such sets is at most
\[
1+\sum_{s=1}^{\ell} 2^{s+|\mathcal B|}(eK)^s
= O\left(\mathrm{poly}(\varepsilon^{-1})\right),
\]
since $|\mathcal B|+\ell=O(\log \varepsilon^{-1})$. By Eq~\ref{eq:brute-runtime}, the marginal on $\mathcal B\cup T$ can be computed in time that is exponential in $|\mathcal B\cup T| = O(\log \varepsilon^{-1})$ and thus its runtime is polynomial: $O\left(\mathrm{poly}(L_{\mathrm{in}},\varepsilon^{-1})\right)$. Finally the cluster expansion on $T'$ can be computed in time $O\left(\mathrm{poly}(L_{\mathrm{in}},\varepsilon^{-1} )\right))$ by Theorem~\ref{thm:MW}. Putting this together, we have that the total runtime is $O\left(\mathrm{poly}(L_{\mathrm{in}},\varepsilon^{-1}) \right)$.
\end{proof}

\section*{Acknowledgments}

MC and SJST thank Nolan Coble, Suchetan Dontha, and Lev Stambler for helpful discussions in the past. MJG, JN, and JR thank Bobak Kiani, Chao Yin, and Yifan (Frank) Zhang for introducing cluster expansion and similar ideas to them for a different project.
MJG acknowledges support from NSF QLCI grant OMA-2120757, NSF QLCI grant OMA-2553574, and NSF NQVL:Design:ORAQL grant 2533041. SJST acknowledges funding support from the NUS Development and is grateful to the University of Maryland's Research Experience for Undergraduates (REU) Combinatorics, Algorithms, and AI for Real Problems (CAAR) for introducing this interesting research question and providing a wonderful research experience back in summer 2021.
\\
\\
\textbf{AI Statement:} The authors used ChatGPT 5.6 Sol to make important connections from ideas regarding partitioning light cones and cluster expansion that the authors have used in their recent papers to formulate proof ideas to solve this research question that they have been thinking about since 2021. ChatGPT and Claude were also used for looking up references and proofreading the manuscript. The authors take full responsibility for the correctness of the manuscript.

\bibliographystyle{alpha}

\bibliography{ref}
\appendix
\section{Fine-grained time complexity}
\label{app:time}

Here we keep $d,K\le k^{2d}$ explicit and compute the runtime. 
\begin{proof}[Proof of the time complexity statement in
Theorem~\ref{thm:algorithm-main} keeping $d,K\le k^{2d}$ explicit]
For $K\ge2$, write
\[
 C_K=\frac{K+1}{\tau}
 =(K+1)[e^2(2K+1)]^{K+1}.
\]
Reading the circuit and constructing $G$ take
$\poly(n,L_{\rm in},d,k^d)$ time.
If $K\le1$, there are at most $n$ $\mathsf{Brute}$ calls
on at most two output bits each, followed by at most $n$
multiplications, so this branch has cost
$O(ndk^d\,2^{2k^d+k})$, in addition to preprocessing.
For $K\ge2$, the $n$ $\mathsf{Brute}$ calls for the exact
computation of the singleton probabilities have cost
$O(ndk^d\,2^{k^d+k})$.

On the branch where we cannot simply output zero,
Section~\ref{sec:core} bounds the size of the set of bad bits by
\[
 |\mathcal B|
 <\frac{K+1}{\tau}\log(1/\varepsilon)
 =C_K\log(1/\varepsilon).
\]
The branch that considers the case where all output bits are
good takes time $T_{\mathsf{Good}}([n],\varepsilon)$, with the
explicit bound given below.
In the branch where all output bits are bad, we have
$n=|\mathcal B|$, and the exact light-cone calculation in \cite{BGM} gives
\[
 T_{\mathsf{Brute}}(n)
 =O(dk^d n\,2^{k^d n+k})
 \le O(dk^d n\,2^k)
       (1/\varepsilon)^{O(k^d C_K)}.
\]

Finally, in the branch where the bits are neither all good
nor all bad, set $R=|\mathcal B|+\ell$, where
$\ell=|\mathcal B|+\lceil\log_2(2/\varepsilon)\rceil$.
Then
\[
 R\le(2C_K+3)\log(2/\varepsilon).
\]
By an argument in Section~\ref{sec:core}, the number of $T$s
that we did not truncate, including the empty set, is
\[
 |\mathcal A|
 \le1+\sum_{s=1}^{\ell}2^{s+b-1}(eK)^s
 \le(2/\varepsilon)^{O(C_K\log(K+2))}.
\]
These $T$s can also be enumerated as follows.
For enumeration only, add edges forming a path through all
vertices of $\mathcal B$. The resulting graph has maximum
degree at most $K+2$. Since every connected component of
$G[T]$ is adjacent to $\mathcal B$, the set
$\mathcal B\cup T$ is connected in this enlarged graph.

Fix a vertex in $\mathcal B$. Starting from this vertex,
enumerate all sequences of at most $2(R-1)$ moves along edges
of the enlarged graph, allowing vertices to be revisited.
Every connected set of $r\le R$ vertices containing the
starting vertex can be visited in at most $2(r-1)$ moves
without leaving that set. Thus every admissible
$\mathcal B\cup T$ occurs as the set of vertices visited by
one of these sequences.

For each sequence, let $V_{\mathrm{seq}}$ be its set of visited
vertices. Keep it only if $\mathcal B\subseteq V_{\mathrm{seq}}$
and $T=V_{\mathrm{seq}}\setminus\mathcal B$ satisfies
$|T|\le\ell$ and every connected component of $G[T]$ is adjacent
to $\mathcal B$ in the original graph. Remove duplicate sets
$T$. Since each move has at most $K+2$ choices, the total cost is
\[
 \poly(n,R)\sum_{j=0}^{2(R-1)}(K+2)^j
 \le
 \poly(n,R)
 (2/\varepsilon)^{O(C_K\log(K+2))},
\]
using $R\le(2C_K+3)\log(2/\varepsilon)$.
Note that the added edges are used only for enumeration.
Every probability calculation still uses the original graph $G$.

For each $T$ that was not truncated, the $\mathsf{Brute}$
call acts on $s=|\mathcal B|+|T|\le R$ outputs.
Eq.~\eqref{eq:brute-runtime} tells us it takes time
\[
 T_{\mathsf{Brute}}(s)
 \le O(dk^d R\,2^{k^dR+k})
 \le O(dk^d R\,2^k)
       (2/\varepsilon)^{O(k^d C_K)}.
\]
Computing $T'$ costs polynomial time. The $\mathsf{Good}$
call uses
\[
 \eta=\frac{\varepsilon}{2(M+1)},
 \qquad
 \eta^{-1}=\frac{2(M+1)}{\varepsilon}
 \le(2/\varepsilon)^{O(C_K\log(K+2))}.
\]

Its runtime cost can be made explicit as follows.
The required marginals are supplied by $\mathsf{Brute}$ calls
with $O(dk^d|S|\,2^{k^d|S|+k})$ cost.
At our choice of $\tau$, the logarithmic truncation estimate
of~\cite{mann2025approximate} is at most $ne^{-m}$ at order $m$. Thus
\[
 m=\left\lceil\log\frac{2n}{\eta}\right\rceil
\]
suffices for relative error at most $\eta$.
An implementation of this same logarithmic truncation gives
\[
 T_{\mathsf{Good}}(W,\eta)
 \le
 \poly(n,L_{\rm in},d,k^d,K,m)\,2^k
 \left[2^{k^d+1}(K+1)^2\right]^m
\]
for every $W\subseteq\mathcal G$.

To justify this computational bound, define
\[
 F_S(z)=\sum_{A\subseteq S}(-z)^{|A|}P(1_A),
 \qquad
 h_U(z)=\sum_{S\subseteq U}
       (-1)^{|U|-|S|}\log F_S(z),
\]
where logarithms are formal power series at zero.
Inclusion--exclusion gives
$\log F_W(z)=\sum_{U\subseteq W}h_U(z)$.
Independence implies that $h_U=0$ if $G[U]$ is disconnected.
Moreover, $h_U$ has no term of degree less than $|U|$:
a degree-$j$ term involves a union of at most $j$ vertices,
and the alternating sum cancels it unless that union is $U$.
It therefore suffices to enumerate connected sets $U$ of
size at most $m$.

As in the preceding enumeration argument, walks of length
at most $2(m-1)$ starting at each vertex give at most
$n(K+1)^{2m}$ candidates. Discard visited sets of size
greater than $m$ and remove duplicates.
For each remaining $U$, evaluating all $P(1_A)$ with
$A\subseteq U$ costs at most
\[
 O\!\left(dk^d m\,2^k(1+2^{k^d})^m\right).
\]
Constructing all $F_S$ directly costs
$3^m\poly(m)$ further operations, and computing their formal
logarithms through degree $m$ costs $2^m\poly(m)$ operations.
These costs give the asserted bound on
$T_{\mathsf{Good}}(W,\eta)$.

Since
\[
 m=O\!\left(
       \log n+
       C_K\log(K+2)\log(2/\varepsilon)
      \right),
\]
each $\mathsf{Good}$ call has cost at most
\[
 \begin{split}
 &\poly(n,L_{\rm in},d,k^d,K,m)\,2^k\,
  n^{O(k^d+\log(K+2))}\\
 &\qquad{}\times
 \left(\frac{2}{\varepsilon}\right)^{
   O\left(C_K(k^d+\log(K+2))\log(K+2)\right)}.
 \end{split}
\]
The branch where all output bits are good satisfies this
bound by taking $\eta=\varepsilon$.

Finally, including all the arithmetic costs that are dependent
on $|\mathcal A|$ gives, for $K\ge2$, total cost at most
\[
 \begin{split}
 &\poly(n,L_{\rm in},d,k^d,K)\,2^k\,
  n^{O(k^d+\log(K+2))}\\
 &\qquad{}\times
 \left(\frac{2}{\varepsilon}\right)^{
   O\left(C_K(k^d+\log(K+2))\log(K+2)\right)}.
 \end{split}
\]
For fixed $d,k$, both $K\le k^{2d}$ and $C_K$ are bounded
constants, so every branch has total cost
$\poly_{d,k}(n,L_{\rm in},1/\varepsilon)$.
\end{proof}

\end{document}